\documentclass[10pt,twoside]{IEEEtran}
\usepackage {amssymb,rotating,graphicx,cite, calc, color, psfrag}
\usepackage{authblk}

\usepackage[cmex10]{amsmath}
\usepackage{amssymb,amsthm,amsfonts,mathrsfs}
\usepackage{graphicx}
\usepackage{enumitem,stackrel}
\usepackage{color}
\usepackage{mathtools,amssymb,bm}
\usepackage{amstext}
\usepackage{array}
\usepackage{algorithmicx}
\usepackage{hyperref}
\theoremstyle{remark}

\ifCLASSOPTIONcaptionsoff
  \usepackage[nomarkers]{endfloat}
 \let\MYoriglatexcaption\caption
 \renewcommand{\caption}[2][\relax]{\MYoriglatexcaption[#2]{#2}}
\fi
\newcommand{\RN}[1]{%
	\textup{\uppercase\expandafter{\romannumeral#1}}%
}

\newtheoremstyle{mystyle}
  {}
  {}
  {\itshape}
  {}
  {\bfseries}
  {.}
  { }
  {}
\theoremstyle{mystyle}

\newtheorem{theorem}{Theorem}{}
\newtheorem{proposition}{Proposition}{}
{}
{}
\newtheorem{defi}{Definition}{}
\newtheorem{remark}{Remark}{}
\usepackage{epstopdf}

\begin{document}
%
\title{Binary Deterministic Sensing Matrix Construction Using Manifold Optimization}

\author{Mohamad Mahdi~Mohades, Hossein~Mohades and S. Fatemeh Zamanian \thanks{M. M. Mohades and S. F. Zamanian are with the School of Electrical Engineering, Shahrekord University, Shahrekord-Iran (e-mail: mohammadmahdi.mohades@sku.ac.ir; zamanian73@gmail.com)\\H. Mohades is with School of Mathematics, Institute for Research in Fundamental Sciences (IPM),
P.O. Box 19395-5746, Tehran - Iran; and University of Farhangian, Tehran - Iran (e-mail:  hoseinmohades@ipm.ir).}}%

\maketitle

\begin{abstract}
Binary deterministic sensing matrices are highly desirable for sampling sparse signals, as they require only a small number of sum-operations to generate the measurement vector. Furthermore, sparse sensing matrices enable the use of low-complexity algorithms for signal reconstruction.
In this paper, we propose a method to construct low-density binary deterministic sensing matrices by formulating a manifold-based optimization problem on the statistical manifold. The proposed matrices can be of arbitrary sizes, providing a significant advantage over existing constructions. We also prove the convergence of the proposed algorithm.
The proposed binary sensing matrices feature low coherence and constant column weight. Simulation results demonstrate that our method outperforms existing binary sensing matrices in terms of reconstruction percentage and signal to noise ratio (SNR).
\end{abstract}


%
\IEEEpeerreviewmaketitle

\section{Introduction}\label{Introduction}
\IEEEPARstart{C}{ompressive} sensing (CS) is a technique for sampling sparse signals at rates below the Nyquist-Shannon rate \cite{Donoho_1}. Sparsity implies that the number of nonzero elements in a signal is significantly smaller than the total number of elements. For reliable sampling that ensures unique and exact reconstruction, specific conditions must be met. Let ${\bf{y}}_{m\times 1}={\bf{A}}_{m\times n}{\bf{x}}_{n\times 1}$, where ${\bf{y}}\in{\mathbb{R}}^{m\times 1}$ represents the measurement vector, ${\bf{A}}\in{\mathbb{R}}^{m\times n}$ is the sampling operator, and ${\bf{x}}\in{\mathbb{R}}^{n\times 1}$ is a sparse signal with $m\ll n$. If ${\bf{x}}$ contains $k$ nonzero elements, it is referred to as a $k$-sparse signal. 

The goal of CS is to reconstruct ${\bf{x}}$ from ${\bf{y}}$ and ${\bf{A}}$, identifying the conditions under which unique and exact reconstruction is achievable. The optimization problem for reconstructing ${\bf{x}}$ is\cite{Candes_2}:
\begin{equation}\label{Norm_0}
\mathop {\min }\limits_{{\bf{z}} \in {^n}} \,\,\,\,{\left\| {\bf{z}} \right\|_0}\,\,\,\,\,\,{\rm{s}}{\rm{.t}}{\rm{.}}\,\,\,\,\,\,\,{\bf{y}} = {\bf{Az}},
\end{equation}
where ${\left\| \cdot \right\|_0}$ denotes the $l_0$-norm, counting the number of nonzero elements. It is proven that Problem (\ref{Norm_0}) has a unique solution if the sampling matrix ${\bf{A}}$ satisfies the null space property (NSP). NSP asserts that no $2k$-sparse signal can reside in the null space of ${\bf{A}}$, where $k$ is the sparsity level of ${\bf{x}}$.
Since Problem (\ref{Norm_0}) is NP-hard \cite{Ge_3}, an alternative optimization problem has been proposed for sparse signal reconstruction \cite{Candes_4}:
\begin{equation}\label{Norm_1}
\mathop {\min }\limits_{{\bf{z}} \in \mathbb{R}^n}  \,\,\,\, {\left\| {\bf{z}} \right\|_1} \,\,\,\,\rm{s.t.}  \,\,\,\, {\bf{y}} = {\bf{Az}},
\end{equation}
where ${\left\| \cdot \right\|_1}$ represents the $l_1$-norm. It has been shown that if ${\bf{A}}$ satisfies the restricted isometry property (RIP), the sparse signal ${\bf{x}}$ can be reconstructed uniquely and exactly by solving Problem (\ref{Norm_1}).

\begin{defi}
The RIP of order $k$ with constant ${{\delta }_{k}}\in [0,1)$ is satisfied for a matrix ${{\mathbf{A}}_{m\times n}}$ if, for all $k$-sparse vectors ${{\mathbf{x}}_{n\times 1}}$, the following inequality holds \cite{Candes_4}:
\begin{equation}\label{Eqn:RIP}
1-{{\delta }_{k}} \leq \frac{\left\| {{\mathbf{A}}_{m\times n}}{{\mathbf{x}}_{n\times 1}} \right\|_{2}^{2}}{\left\| {{\mathbf{x}}_{n\times 1}} \right\|_{2}^{2}} \leq 1+{{\delta }_{k}}.
\end{equation}
\end{defi}
  In \cite{Baraniuk_6}, it is demonstrated that Gaussian matrices, satisfy the RIP property with high probability. However, the primary drawback of random sensing matrices is the need for element-by-element storage, which demands significant memory resources. This limitation motivates the use of deterministic structures, which only require storing generating functions instead of individual elements.
Verification of the RIP property for deterministic sensing matrices is often challenging. Instead, matrix coherence, defined as follows, is used as a surrogate measure. 
\begin{defi}\label{Coherence}
The coherence of a matrix ${{\mathbf{A}}_{m\times n}}$ with columns $\left[ {{\mathbf{a}}_{1}},\cdots ,{{\mathbf{a}}_{n}} \right]$ is given by:
\begin{equation}\label{Coherence_Equ}
{{\mu }_{\mathbf{A}}} = \underset{i\ne j}{\mathop{\max }}\,\frac{\left| \left\langle {{\mathbf{a}}_{i}},{{\mathbf{a}}_{j}} \right\rangle \right|}{{{\left\| {{\mathbf{a}}_{{{i}_{{}}}}} \right\|}_{2}} \cdot {{\left\| {{\mathbf{a}}_{j}} \right\|}_{2}}}, \quad 1\le i,j\le n.
\end{equation}
The minimum attainable coherence for a given matrix ${{\mathbf{A}}_{m\times n}}$ is bounded below by the Welch bound \cite{Welch}:
\begin{equation}\label{Welch}
{{\mu }_{\mathbf{A}}}\ge \sqrt{\frac{n-m}{m\left( n-1 \right)}}.
\end{equation}
\end{defi}

In \cite{bourgain}, it is shown that a matrix ${{\mathbf{A}}}$ with coherence ${{\mu }_{\mathbf{A}}}$ satisfies RIP of order $k$ with constant $\delta_k\le {{\mu }_{\mathbf{A}}}(k-1)$ for $k<\frac{1}{{{\mu }_{\mathbf{A}}}}+1$. Thus, to achieve RIP for higher orders, it is necessary to design deterministic sensing matrices with low coherence \cite{tong2021deterministic, mohades2016non}.

An important class of deterministic sensing matrices is binary constructions, which enable multiplier-less dimensionality reduction. Moreover, binary sensing matrices often lead to efficient algorithms due to their sparse structure \cite{Indyk}. Many deterministic structures are based on finite fields. 

In a seminal work by DeVore \cite{DeVore}, binary deterministic sensing matrices of size $p^2 \times p^{r+1}$ and coherence $\frac{r}{p}$ were proposed using polynomials of degree $r$ whose coefficients belong to the finite field $\mathbb{F}_p$, where $p$ is a prime number. Inspired by DeVore's structure, further binary constructions were introduced in \cite{Ali, gan2022non}. Binary sensing matrices with block structures can be found in \cite{Euler}. Additionally, extremal set theory has been used to construct binary sensing matrices, as shown in \cite{Extremal}.
The authors of \cite{xia2014deterministic} developed binary deterministic sensing matrices based on LDPC codes. Similar constructions can be found in \cite{li2014deterministic}, where finite geometry was utilized to construct sparse sensing matrices.
Other binary deterministic sensing matrices constructed based on LDPC codes can be found in \cite{haiqiang}. In \cite{mohades2019general}, a general approach for constructing binary deterministic sensing matrices was proposed using column replacement. 

Apart from binary constructions, there exist a wide variety of deterministic sensing matrices. In \cite{MohadesTensor} and \cite{abin2023explicit}, some asymptotically optimal complex-valued deterministic sensing matrices, in terms of the Welch bound, were proposed. However, complex-valued structures require significantly more operations for dimensionality reduction compared to binary matrices, making them less practical in many cases.

The introduced matrices face a significant limitation of having special sizes, restricting their applicability and performance across diverse applications. Even the method in \cite{mohades2019general}, which combines different matrices to construct sensing matrices, suffers from the same limitation.

Overcoming this limitation serves as the primary motivation for proposing an optimization-based method to construct deterministic sensing matrices. Additionally, achieving multiplier-less dimensionality reduction and exploiting the sparse structure of binary sensing matrices, which facilitates low-complexity algorithms \cite{Indyk}, further motivates the development of binary sensing structures.

To generate binary deterministic sensing matrices, we propose an optimization problem focusing on minimizing the Hamming distance between the columns of the sensing matrix. The optimization problem enforces the following constraints on the sensing matrix:
\begin{enumerate}
    \item The matrix must have binary values.
    \item The matrix must exhibit low coherence.
    \item The matrix must be column-regular.
\end{enumerate}
The latter constraint encourages working on manifolds, as will be discussed in detail in subsequent sections.

Our approach for constructing binary matrices offers the following advantages over other methods:
\begin{enumerate}
    \item[$\rm{i}$] The matrices can have arbitrary sizes, overcoming the dimensional limitations of other methods.
    \item[$\rm{ii}$] The columns of our matrices are not constrained by the linearity of codes. Specifically, most deterministic sensing matrices, particularly binary ones, are constructed from matrices whose columns correspond to linear codewords, which restricts the number of sensing matrix columns. 
    \item[$\rm{iii}$] Our proposed method allows flexibility in manipulating the sensing matrices, including enhancing sparsity by penalizing the number of nonzero elements in the rows.
\end{enumerate}
After formulating the optimization problem, we present an iterative algorithm for solving it and provide a convergence proof.\\
The remainder of the paper is organized as follows. In Section \ref{Main_results}, we present our optimization formulation, provide analytical discussions, and prove relevant theorems. Simulation results are illustrated in Section \ref{Simulation}. We conclude the paper in Section \ref{Conclusion}. Additionally, in Appendix \ref{Appendix}, we provide the required preliminaries used in explaining our main results.
\section{Main Result} \label{Main_results}

As highlighted in Section \ref{Introduction}, an effective deterministic sensing matrix requires minimal coherence. Additionally, for the binary construction under consideration, all columns must have an identical number of nonzero elements. Based on these criteria, we propose the following optimization problem:
\begin{equation} \label{Proposed_1}
\begin{aligned}
& \underset{{\mathbf{a}_i}, {\mathbf{a}_j} \in \mathbb{R}^m}{\min} \,\, \underset{i \neq j}{\max} \,\, \frac{\left| \langle \mathbf{a}_i, \mathbf{a}_j \rangle \right|}{\|\mathbf{a}_i\|_2 \|\mathbf{a}_j\|_2}, \\
& \text{s.t.}
\,\,\,\,\,\,\text{C1: } \mathbf{a}_l \odot (\mathbf{a}_l - \mathbf{1}) = \mathbf{0}, \,\, l = 1, \ldots, n, \\
&\,\,\,\,\,\,\,\,\,\,\,\,\, \text{C2: } \sum_k \frac{a_{kl}}{r} = 1, \,\, l = 1, \ldots, n,
\end{aligned}
\end{equation}
where \( i, j = 1, \ldots, n \), \( k = 1, \ldots, m \), \( a_{ki} \) denotes the \( k \)-th element of \(\mathbf{a}_i\), the \( i \)-th column of matrix \(\mathbf{A}\), \( r \) is the column weight, \(\|\cdot\|_2\) represents the \( \ell_2 \)-norm, \(\langle \cdot, \cdot \rangle\) is the inner product, and \(\odot\) denotes the element-wise product. 

The goal of Problem (\ref{Proposed_1}) is to minimize the coherence of matrix \(\mathbf{A}\). Condition \(\text{C1}\) enforces binary matrix elements, while \(\text{C2}\) ensures column regularity. Notably, \(\text{C2}\) implies that each column can be interpreted as a probability distribution (after normalization by \(r\)). Considering this, the problem can be reformulated over a statistical manifold\ref{StatisticalManifold}, leading to:
\begin{equation} \label{Proposed_2}
\begin{aligned}
& \underset{\frac{\mathbf{a}_i}{r}, \frac{\mathbf{a}_j}{r} \in S_m}{\min} \,\, \underset{i \neq j}{\max} \,\, \frac{\left| \langle \mathbf{a}_i, \mathbf{a}_j \rangle \right|}{\|\mathbf{a}_i\|_2 \|\mathbf{a}_j\|_2}, \\
& \text{s.t.}\,\,\,\,\, \text{C1: } \mathbf{a}_l \odot (\mathbf{a}_l - \mathbf{1}) = \mathbf{0}, \,\, l = 1, \ldots, n.
\end{aligned}
\end{equation}
Relaxing \(\text{C1}\) allows for a transition from discrete to continuous space, facilitating optimization via differential geometry:
\[
\text{C1}': \mathbf{a}_l^T (\mathbf{a}_l - \mathbf{1}) = 0.
\]
Consequently, Problem (\ref{Proposed_2}) transforms into:
\begin{equation} \label{Proposed_3}
\begin{aligned}
& \underset{\frac{\mathbf{a}_i}{r}, \frac{\mathbf{a}_j}{r} \in S_m}{\min} \,\, \underset{i \neq j}{\max} \,\, \frac{\left| \langle \mathbf{a}_i, \mathbf{a}_j \rangle \right|}{\|\mathbf{a}_i\|_2 \|\mathbf{a}_j\|_2}, \\
& \text{s.t.} \,\,\,\,\,\, \text{C1}': \mathbf{a}_l^T (\mathbf{a}_l - \mathbf{1}) = 0, \,\, l = 1, \ldots, n.
\end{aligned}
\end{equation}
To ensure smooth optimization, the \(\max\) function in (\ref{Proposed_3}) is approximated using the smooth function \(M_\alpha\):
\[
M_\alpha(\mathbf{x}) = \frac{\sum_{i=1}^m x_i e^{\alpha x_i}}{\sum_{i=1}^m e^{\alpha x_i}}, \quad \text{where } \mathbf{x} = (x_1, \ldots, x_m),
\]
with \( M_\alpha(\mathbf{x}) \to \max(\mathbf{x}) \) as \(\alpha \to \infty\). Using this approximation, the smooth reformulation of (\ref{Proposed_3}) becomes:
\begin{equation} \label{Proposed_4}
\begin{aligned}
& \underset{\mathbf{b}_i, \mathbf{b}_j \in S_m}{\min} \,\, \frac{\sum_{i \neq j} r^2 \gamma_{ij} e^{\alpha r^2 \gamma_{ij}}}{\sum_{i \neq j} e^{\alpha r^2 \gamma_{ij}}}, \\
& \text{s.t.}\,\,\,\,\,\,\, \text{C1}': \mathbf{b}_l^T (\mathbf{b}_l - \mathbf{1}/r) = 0, \,\, l = 1, \ldots, n,
\end{aligned}
\end{equation}
where \(\mathbf{b}_i = \mathbf{a}_i / r\), \(\mathbf{1}/r\) is a vector with elements equal to \(1/r\), and \(\gamma_{ij} = \langle \mathbf{b}_i, \mathbf{b}_j \rangle\). Due to \(\mathbf{b}_l \in S_m\), \(\text{C1}'\) simplifies to:
\begin{equation} \label{Condition_C1}
\text{C1}'': \mathbf{b}_l^T \mathbf{b}_l = 1/r, \,\, l = 1, \ldots, n.
\end{equation}
Thus, Problem (\ref{Proposed_4}) is reformulated as:
\begin{equation} \label{Proposed_5}
\begin{aligned}
& \underset{\mathbf{b}_i, \mathbf{b}_j \in S_m}{\min} \,\, \frac{\sum_{i \neq j} r^2 \gamma_{ij} e^{\alpha r^2 \gamma_{ij}}}{\sum_{i \neq j} e^{\alpha r^2 \gamma_{ij}}}, \\
& \text{s.t.}\,\,\,\,\,\,\, \text{C1}'': \mathbf{b}_l^T \mathbf{b}_l = 1/r, \,\, l = 1, \ldots, n.
\end{aligned}
\end{equation}
To reformulate Problem (\ref{Proposed_5}) as an unconstrained manifold optimization problem, we employ Theorem~\ref{Submersion_Theorem} and propose the following proposition.
\begin{proposition}\label{Proposition_Ours_1}
Assume that the set $ES_m$ contains those points of $S_m$ which satisfy the condition $C1''$ of Problem (\ref{Proposed_5}), {\it{i.e.}} $\{{\bf{x}}\in ES_m| {\bf{x}}\in S_m\,\,\, and \,\,\,{\bf{x}}^T{\bf{x}}={{\frac{1}{r}}}\}$. Then, $ES_m$ is an embedded submanifold of $S_m$ with dimension of $m-1$. 
\end{proposition}
\begin {proof}
Consider the function $F:S_m\to\mathbb{R}$, where $F({\bf{x}})={\bf{x}}^T{\bf{x}}-{{\frac{1}{r}}}$. By using Definition \ref{Rank}, it can be seen that the rank of the function $F$ is constant and equal to $1$. To see this, let us define the atlases $\mathcal{A}_1$ and $\mathcal{A}_2$ for manifolds $S_m$ and $\mathbb{R}$, respectively. $\mathcal{A}_1$ is a collection of charts $\{(\varphi_i,\mathcal{U}_i)\}$ for $i=0,1,...,m$, where $\varphi_i:\mathcal{U}_i\to D_m$ is defined as $\varphi_i({\bf{x}})=(x_0,x_1,\cdots,x_{i-1},\widehat{x}_i,x_{i-1},\cdots,x_m)$, where $\widehat{x}_i$ stands for elimination of ${x}_i$ and the open subset $D_m$ is defined as $D_m=\{{\bf{y}}\in \mathbb{R}^m| y_i>0\,\,\,\, \& \,\,\sum_{i=0}^{m-1}y_i<1\}$. Moreover, the mapping $\varphi_i^{-1}:D_m\to \mathcal{U}_i$ is defined as $\varphi_i^{-1}({\bf{y}})=(y_0,y_1,\cdots,y_{i-1},1 - \sum_{i = 0}^{m - 1} {{y_i}},y_{i},\cdots,y_{m-1})$. It is easy to see that both $\varphi_i$ and $\varphi_i^{-1}$ are not only continuous mappings but also infinitely differentiable mappings. Furthermore, $\mathcal{A}_1$ can be obviously presented only using the chart $(\iota,\mathbb{R})$, where $\iota$ is the identity map from $\mathbb{R}$ to $\mathbb{R}$. Accordingly, using defintion \ref{Rank} the function $F:S_m\to\mathbb{R}$ is smooth, if $\hat F = \iota \circ F \circ \varphi_i^{ - 1}=F \circ \varphi_i^{ - 1}:{\mathbb{R}^{{m}}} \to {\mathbb{R}}$ is smooth for any chart $(\varphi_i,\mathcal{U}_i)$, $i=0,1,...,m$. The function $\hat F$ is easily seen to be smooth for any chart $(\varphi_i,\mathcal{U}_i)$, due to it is a polynomial. Moreover, to find the rank of $\hat F$, we know that $\hat F(\varphi({\bf{x}}))={\bf{z}}^T{\bf{z}}-{{\frac{1}{r}}}$ where ${\bf{z}}=(x_0,x_1,\cdots,x_{i-1},1 - \sum_{j = 0 , j \ne i\hfill}^m {{x_j}},x_{i+1},\cdots,x_{m})$. Consequently, the differential of $\hat F$ is $( {x_0^{},x_1^{}, \cdots ,x_{i - 1}^{},{ - 2m + 2m\sum_{{j = 0 , j \ne i\hfill}}^m {{x_j}} } ,x_{i + 1}^{}, \cdots ,x_m^{}} )$, which is a nonzero vector. Therefore, the rank of $\hat F$ is constant and equal to 1. Using the mentioned arguments and applying Theorem \ref{Submersion_Theorem}, it is observed that $ES_m$ is an embedded submanifold of $S_m$ with dimension of $m-1$.
\end{proof}
\begin{center} \label{Line_Search_Absil}
\resizebox{0.45\textwidth}{!}{%
\begin{tabular}{|c|}
\hline
\hspace{-.2cm}\textbf{Alg. 1}: Manifold Gradient descent \\
\hline
\hspace{-.1cm}\textbf{Requirements}: Differentiable cost function $f$, Manifold $\mathcal{M}$,\\ \hspace{-.6cm}inner product $\left<\cdot,\cdot\right>$, Initial matrix ${\bf{B}}_0 \in \mathcal{M}$, Retraction \\ function $R$, Scalars $\bar{\alpha}>0, \, \, \beta,\sigma\in(0,1)$, tolerance $\tau>0$.\\
\hspace{-5.7cm}{\bf{for}} $i=0,1,2,...$ {\bf{do}}\\
\hspace{-.4cm}\textbf{Step 1}: \,\,\,\, Set $\xi$ as the negative direction of the gradient, \\
\hspace{-2.6cm}$\xi_i:=-{\rm{grad}}f({\bf{B}}_i)$\\
\hspace{-3.8cm}\textbf{Step 2}:\,\,\,\, Convergence evaluation,\\
\hspace{-2cm}\,\,\,\, \textbf{if} $\left\|\xi_i\right\| < \tau$, \textbf{then break}\\
\hspace{-3.00cm}\textbf{Step 3}:\,\,\,\, Find the smallest $m$ satisfying\\
\,\,\,\,\,\,\,\, $f({\bf{B}}_i)-f(R_{{\bf{B}}_i}(\bar{\alpha}\beta^m\xi_i))\geq \sigma\bar{\alpha}\beta^m \left<\xi_i,\xi_i\right>_{{\bf{B}}_i}$\\
\hspace{-3.2cm}\textbf{Step 4}:\,\,\,\,\,\,\,\,\, Find the modified point as \\
\,\,\,\,\,\,\,\,\,\, ${\bf{B}}_{i+1}:=R_{{\bf{B}}_i}(\bar{\alpha}\beta^m\xi_i))$\\
\hline
\end{tabular}
}
\end{center}
Using proposition \ref{Proposition_Ours_1}, (\ref{Proposed_5}) can be rewritten as follows,
\begin{equation} \label{Proposed_6}
\begin{array}{l}
\mathop {\min }\limits_{{{\bf{b}}_i},{{\bf{b}}_j}\in {ES}{_m}} f({\bf{B}})=\mathop \frac{{\sum_{i \ne j} {{r^2\gamma _{ij}}{e^{\alpha r^2{\gamma _{ij}}}}} }}{{\sum_{i \ne j} {{e^{\alpha r^2{\gamma _{ij}}}}} }},
\end{array}
\end{equation}
where ${\gamma _{ij}} = \left\langle {{{\bf{b}}_i},{{\bf{b}}_j}} \right\rangle$, $i, j=1,\cdots,n$ and ${\bf{B}}$ is a matrix whose $i$-th column is ${\bf{b}}_i$. Note that ${\bf{B}}$ belongs to the product manifold ${ES}{_m^n}$ which is defined as ${ES}{_m^n}={ES}{_m}\times{ES}{_m}\times\cdots\times{ES}{_m}$, wherein $\times$ stands for the Cartesian product.\\To solve (\ref{Proposed_6}), gradient descent (GD) approach can be applied. Alg. 1 presents a common manifold GD method \cite{Absil_book},\cite{mohades2021efficient},\cite{mohades2019haplotype}. Steps 1 sets the search direction. Step 2 evaluates the convergence criterion. Step 3 guarantees that the sequence of points is a descent direction on the manifold\cite{Absil_book}. Step 4 retracts the updated point to the manifold.\\
Now, the convergence of Alg. \ref{Line_Search_Absil} for (\ref{Proposed_6}) is proved.
\begin {theorem} \label {Convergence_Proof}
By choosing the initial matrix ${\bf{B}}_0$ belonging to the manifold $ES_m^n$, Alg. 1 will converge for (\ref{Proposed_6}).
\end{theorem}
\begin{proof}
The limit point of the infinite sequence $\{{\bf{B}}_i\}$ generated by Alg. 1 is a critical point of $f(\bf{B})$ in (\ref{Proposed_6}); {\it{i.e.}} $\mathop {\lim }\limits_{i \to \infty } {\left\| {{\rm{grad}}f\left( {{{\bf{B}}_i}} \right)} \right\|_F} = 0$ (see Theorem 4.3.1 in \cite{Absil_book}). Therefore, to show the convergence of Alg. 1 for (\ref{Proposed_6}), we should illustrate that $\{{\bf{B}}_i\}$ owns a limit point belonging to the manifold $ES_m^n$. For this purpose, we prove that $\{{\bf{B}}_i\}$ lies in a compact set wherein for any sequence there exists a convergent subsequence \cite{kreyszig}.\\
A compact set in a finite dimension space in which any two distinct points can be separated by disjoint open sets is bounded and closed \cite{kreyszig}. Due to $ES_m^n$ is bounded, the set containing $\{{\bf{B}}_i\}$ is bounded.
Now, we require to show that the set containing $\{{\bf{B}}_i\}$ is closed under the condition of Theorem.
Retraction function $R$ brings back the updated point to the manifold $ES_m^n$ by normalizing the columns of the matrix. This can be done for all bounded matrices except the all-zero matrix. Therefore, we should prove that the zero matrix will not be generated. Assume that the zero matrix ${\bf{0}}_{m \times n}$ is a convergent point, then we should have 
$\left\| {{{\bf{B}}_i} - {{\bf{0}}_{m \times n}}} \right\|_F^2 < \varepsilon $ for some $i>k$, where $\epsilon$ is an arbitrarily small positive value, $k$ is positive integer and the subscript $F$ stands for the Frobenius norm. This is a contradiction, because, the Frobenius norm of any matrix belonging to the manifold $ES_m^n$ is equal to ${\raise0.7ex\hbox{$n$} \!\mathord{\left/
{\vphantom {n {{m^2}}}}\right.\kern-\nulldelimiterspace}
\!\lower0.7ex\hbox{${{r}}$}}$ due to Proposition \ref{Proposition_Ours_1}. Hence, Alg. 1 converges for (\ref{Proposed_6}).
\end{proof}
\begin{remark}
Now, to construct the binary deterministic matrix, the elements of the generated matrix must be mapped to binary values. This is achieved using a simple thresholding method. Since each desired sensing matrix should be column-regular, the $r$ largest values in each column are mapped to $1$, while the remaining values are mapped to $0$.
\end{remark}
\section {Simulation Results} \label{Simulation}
Due to the absence of binary deterministic sensing matrices of arbitrary sizes, in order to illustrate the well performance
of our proposed matrices with respect to the existing binary
sensing matrices, we limit the size of our proposed matrices to specific cases of $p^2 \times p^4$, where $p$ is a prime power. This choice enables comparison with DeVore's construction. Additionally, random binary matrices of the same size are generated for further comparison.
\begin{figure}[htbp]
  \centering
\begin{minipage}[b]{0.4\textwidth}
    \includegraphics[width=\textwidth]{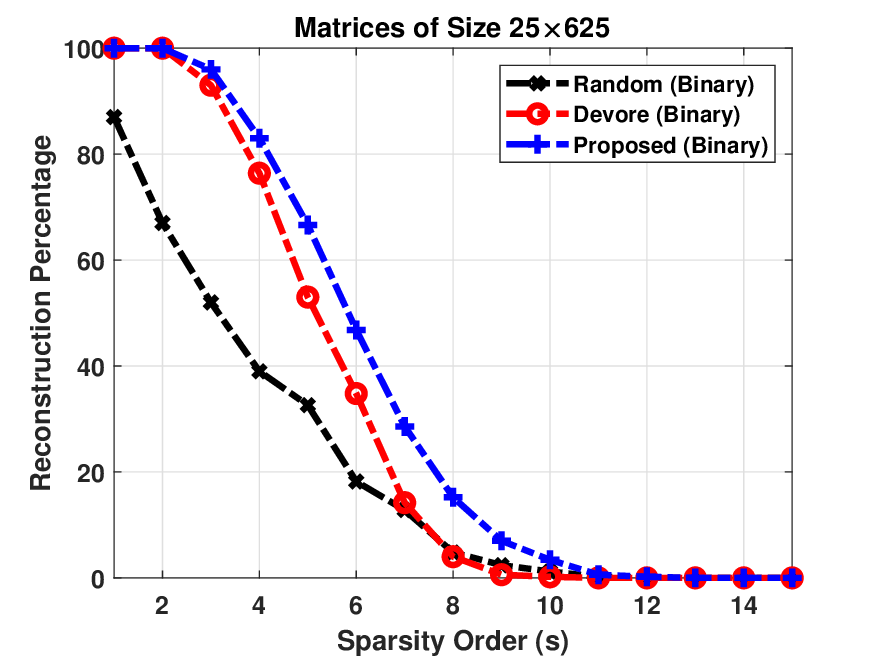}
  \caption{\label{Percentage_Sparsity25625}Reconstruction percentage versus sparsity order.}
 \end{minipage}
\end{figure}
%
The results are derived from the average of 2000 independent trials for various $k$-sparse signals. In each scenario, the performance of the measurement matrices is assessed based on recovery percentage and output SNR, while the sparsity order increases and the input SNR ranges from 0 to 100 dB. The reconstruction algorithm employed is Orthogonal Matching Pursuit, which is effective for solving $\ell_1$-minimization problems.
\begin{figure}[htbp]
  \centering
\begin{minipage}[b]{0.4\textwidth}
 \includegraphics[width=\textwidth]{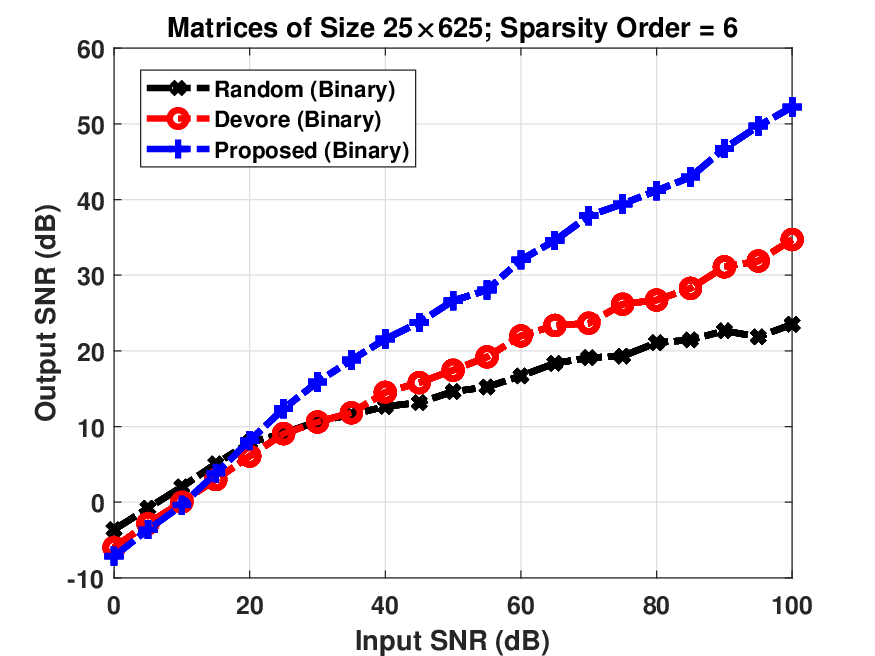}
 \caption{\label{SNR_SNR_25625}Reconstruction SNR versus input SNR.}
  \end{minipage}
\end{figure}
For the initial setup, we set $p = 5$, yielding sensing matrices of size $25 \times 625$. The corresponding results are shown in Figs.~1 and 2. In Fig.~1, the sparsity order ranges from 1 to 15, whereas in Fig.~2, it is fixed at 6. The results demonstrate that the proposed sensing matrix consistently outperforms existing ones in terms of recovery percentage and output SNR.
\begin{figure}[htbp]
  \centering
\begin{minipage}[b]{0.4\textwidth}
 \includegraphics[width=\textwidth]{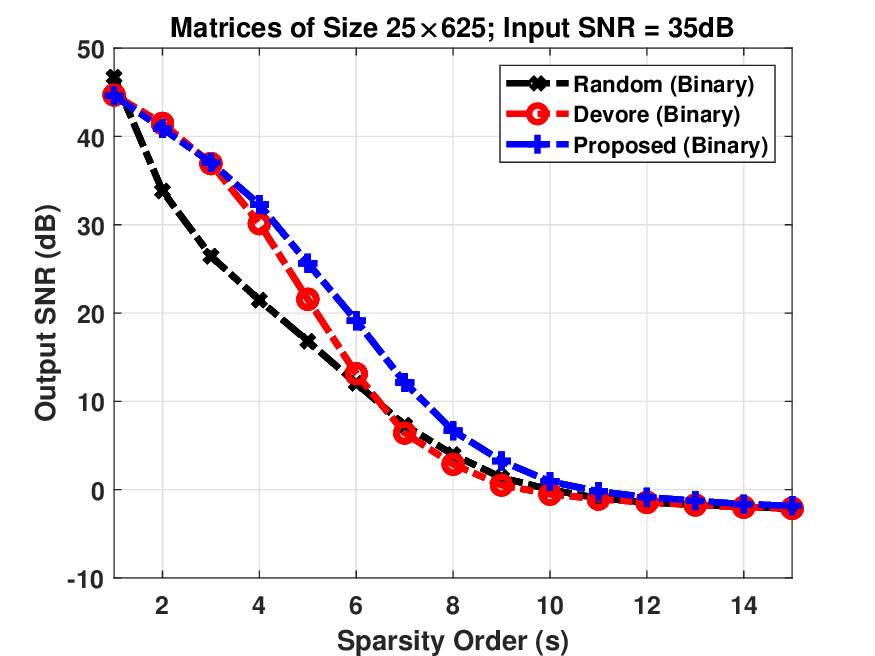}
 \caption{\label{SNR_SNR_25625}Reconstruction SNR versus sparsity order.}
  \end{minipage}
\end{figure}
In the final scenario, the input SNR is held constant at $35\,\mathrm{dB}$, while the sparsity order is varied from 1 to 15 to compute the output SNR. The results, illustrated in Fig.~3, further confirm that the proposed sensing matrices exhibit superior performance compared to other matrices in terms of output SNR.
\section {Conclusion} \label{Conclusion}
Manifold optimization was employed to design low-coherence binary deterministic sensing matrices of arbitrary sizes. To achieve this, an iterative algorithm was developed to minimize the coherence of the sensing matrix over the statistical manifold. The optimization problem was then formulated over the intersection of the statistical manifold and a sphere, with the radius proportional to the number of non-zero elements per column. Due to column regularity, this value is constant across all columns. 
Next, the problem was approximated and reformulated as a smooth optimization problem, solvable using the manifold gradient descent algorithm. The convergence of the proposed algorithm was formally proven. Upon convergence, the optimal solution was used to construct the binary deterministic sensing matrix.

Simulation results demonstrated that the proposed algorithm is not only flexible in generating binary sensing matrices but also reliable for sampling and reconstructing sparse signals, as evidenced by superior recovery percentages and output SNRs.
\renewcommand{\theequation}{A\arabic{equation}}
  \setcounter{equation}{0}  
\appendices
\section{Defintions}  \label{Appendix}
\begin{defi} \label{Tangent_space}
Let $S$ be a manifold.  Then, $T_{{\bf{z}}}{S}$ is the set of all tangent vectors at point ${{\bf{z}}}\in{S}$ and called the tangent space. Moreover, tangent bundle is defined as the disjoint collection of tangent spaces and denoted by $T{S}$\cite{Absil_book}.
\end{defi}
\begin {defi} \label {Chart}
Let $\mathcal{M}$ be a second countable manifold. Then, it contains a collection of charts $\left(\mathcal{U}_\alpha,\varphi_\alpha \right)$, where $\mathcal{U}_\alpha$'s are subsets of $\mathcal{M}$ satisfying $\cup_\alpha\mathcal{U}_\alpha=\mathcal{M}$, and  $\varphi_\alpha$ is a continuous mapping, whose inverse is also continuous, between  $\mathcal{U}_\alpha$  and an open subset of  $\mathbb{R}^d$. Moreover, for any pair $\alpha, \beta $ with $\mathcal{U}_\alpha \cap \mathcal{U}_\beta \neq \emptyset$, the sets $\varphi_\alpha\left(\mathcal{U}_\alpha \cap \mathcal{U}_\beta\right)$ and $\varphi_\beta\left(\mathcal{U}_\alpha \cap \mathcal{U}_\beta\right)$ are open subsets in $\mathbb{R}^d$ and 	${\varphi _\beta } \circ \varphi _\alpha ^{ - 1}:{\mathbb{R}^d} \to {\mathbb{R}^d}$ is infinitely differentiable on its domain, also known as smooth \cite {Absil_book}.
\end{defi}
\begin {defi}  \label{Rank}
Let $F:\mathcal{M}_1\to\mathcal{M}_2$ be a functoin from the $d_1$ dimensional manifold $\mathcal{M}_1$ to the manifold $\mathcal{M}_2$ with dimension of $d_2$. Then, $F$ is differentiable (smooth) at $x$ if $\hat F = {\psi} \circ F \circ \varphi^{ - 1}:{\mathbb{R}^{{d_1}}} \to {\mathbb{R}^{{d_2}}}$ is differentiable (smooth) at $\varphi(x)$, where $(\varphi,\mathcal{U})$ and $(\psi,\mathcal{V})$ are charts around $x$ and $\varphi(x)$, respectively.
Moreover, the rank of $F$ at point $x\in\mathcal{M}$ is defined as the dimension of the image of  the differential of $\hat F$ at $\varphi(x)$ \cite{Absil_book}.
\end{defi}
\begin {theorem} \label{Submersion_Theorem}
Let $F:\mathcal{M}_1\to\mathcal{M}_2$ be a smooth mapping between two manifolds of dimension $d_1$ and $d_2$.
Then, for $y\in F(\mathcal M)$, $F^{-1}(y)$ is called an embedded submanifold of $\mathcal{M}_1$ of dimension  $d_1-k$, provided that $F$ has a constant rank $k<d_1$ in a neighborhood of $F^{-1}(y)$  \cite{Absil_book}.
\end{theorem}
\begin{defi}\label{StatisticalManifold}
Let $ {\mathcal{M}}_{m+1}$ be the set of all positive measures:
\begin{equation}\label{positive_measures:}
{ {\mathcal{M}}_{m + 1}} = \left\{ {\left. {{\bf{p}} = \left( {{p_0},...,{p_m}} \right)} \right|0 < {p_i},i = 0,...,m} \right\}.
\end{equation}
Then, the subspace $S_{m}$ of $ {\mathcal{M}}_{m+1}$ with property $\sum\limits_{i = 0}^m {{p_i}}  = 1$ is called statistical manifold\cite{amari2009divergence}. 
\end{defi}
\section*{Acknowledgement}  \label{Ack}
This research was in part supported by a grant from IPM (No.1400510023).
\bibliographystyle{ieeetr}
\bibliography{Binary_matrix_References_030902}
\end{document}